\documentclass[11pt,a4paper]{article}
\usepackage[margin=1in]{geometry}
\usepackage{graphicx}
\usepackage{array}
\usepackage{eqnarray}
\usepackage{latexsym}
\usepackage{subfig}
\usepackage{theorem}
\usepackage{euler}

{\theorembodyfont{\itshape}
 \newtheorem{theorem}{Theorem}}

{\theorembodyfont{\itshape}
 \newtheorem{definition}{Definition}}
 
 {\theorembodyfont{\itshape}
 }

{\theorembodyfont{\itshape}
 \newtheorem{corollary}{Corrolary}}

{\theorembodyfont{\itshape}
 \newtheorem{proposition}{Proposition}}

{\theorembodyfont{\itshape}
 \newtheorem{lemma}{Lemma}}

{\theorembodyfont{\itshape}
 \newtheorem{remark}{Remark}}

\newcommand{\dam}{{\mbox{\rm dam}}}
\newcommand{\impr}{{\mbox{\rm impr}}}

\newcommand{\qed}{\hfill$\Box$}

\title{
On Pure and (approximate) Strong Equilibria\\ of Facility Location Games
}

\author{
Thomas Dueholm Hansen$^1$
\mbox{\ \ \ \ }
Orestis A. Telelis$^1$
\medskip\\
Department of Computer Science, Aarhus University, Denmark\\
\texttt{\{tdh, telelis\}@cs.au.dk}
}

\date{}

\begin{document}
\bibliographystyle{plain}
\maketitle
\footnotetext[1]{
Member of The Center for Algorithmic Game Theory, funded by the Carlsberg Foundation, Denmark.
}

\begin{abstract}
We study social cost losses in Facility Location games, where $n$ selfish
agents install facilities over a network and connect to them, so as to forward
their local demand (expressed by a non-negative weight per agent). Agents
using the same facility share fairly its installation cost, but every agent
pays individually a (weighted) connection cost to the chosen location. We
study the Price of Stability (PoS) of {\em pure} Nash equilibria and the Price
of Anarchy of {\em strong} equilibria (SPoA), that generalize pure equilibria
by being resilient to coalitional deviations. A special case of recently
studied network design games, Facility Location merits separate study as a
classic model with numerous applications and individual characteristics: our
analysis for unweighted agents on metric networks reveals constant upper and
lower bounds for the PoS, while an $O(\ln n)$ upper bound implied by previous
work is tight for non-metric networks. Strong equilibria do not always exist,
even for the unweighted metric case. We show that $e$-approximate strong
equilibria exist ($e=2.718\dots$). The SPoA is generally upper bounded by
$O(\ln W)$ ($W$ is the sum of agents' weights), which becomes tight
$\Theta(\ln n)$ for unweighted agents. For the unweighted metric case we prove
a constant upper bound. We point out several challenging open questions that
arise.
\end{abstract}

\section{Introduction}
\label{section:introduction}

Modern computer and communications networks constitute an arena of economic
interactions among multiple autonomous self-interested entities (network
access providers, end-users, electronic commerce enterprises, content storage
and distribution enterprises). The internet is perhaps the most massive and
global field where economic networking interactions take place. The recently
established study of {\em network formation games}~\cite{agt}(chapter 19) aims
at understanding how the competitive (or coalitional) activity of multiple
such {\em selfish agents} affects the network's characteristics and its
efficiency. In this paper we consider the setting of $n$ {\em Content
Distribution Network} enterprises interacting over a network. Enterprise $i$
takes decisions on where to store replicas of digital content over the
network, so as to satisfy access demand of local customers, situated at node
$u_i$. Demand is expressed by a non-negative weight $w_i$ for enterprise $i$.
Every enterprise chooses strategically a location $v$ on the network for
installation of content replicas, so as to minimize its individual expenses
for {\em (a)} storage/management of the content, and {\em (b)} weighted
connection (bandwidth/delay) costs to the chosen location. We study a {\em
Facility Location} game played among selfish agents (enterprises), that
naturally models this situation. We use a fair cost allocation rule - known as
Shapley cost-sharing~\cite{Anshelevich04} - for facility installation costs:
for a location $v$ installation cost $\beta_v$ is shared in a fair manner
among agents that receive service from $v$, so that agent $i$ pays
$\frac{w_i\beta_v}{W(v)}$, where $W(v)$ denotes the total demand forwarded to
$v$. The weighted connection cost is given by $w_id(u_i,v)$, where $d(u_i,v)$
denotes connection cost per unit of demand.

Our work focuses on bounding the {\em social cost} of cost-efficient and {\em stable} network infrastructures; stable networks will be represented by {\em pure} strategy Nash equilibria and {\em strong} equilibria of the corresponding Facility Location game. The social cost will be the sum of individual costs experienced by the agents-players of the game. Strong equilibria are an extension of pure Nash equilibria, dicussed in~\cite{Andelman07,Epstein07}, and essentially introduced by Aumann in~\cite{Aumann59}. Strong equilibria are resilient to coalitional {\em pure} deviations: a strategy profile is a strong equilibrium if no subset of agents can deviate (by jointly adopting a different {\em pure} strategy), so that all of its members are better off. Thus this notion captures the possibility of coalitional behavior among agents. Occurence of such behavior is naturally expected in rapidly evolving modern markets, especially the ones involving the exploitation and exchange of digital goods and services.

We derive bounds on the {\em Price of Stability} (PoS) of pure equilibria,
defined as the cost of the least expensive equilibrium relative to the
socially optimum cost~\cite{Anshelevich04}. For strong equilibria we derive
bounds on their (strong) {\em Price of Anarchy} (SPoA), i.e. the cost of the
most expensive strong equilibrium relative to the socially optimum cost. Let
us note that the {\em Price of Anarchy} of pure Nash equilibria was introduced
in~\cite{Koutsoupias99} as the cost of the most expensive pure Nash
equilibrium relative to the socially optimum cost, and measures essentially
social cost losses incurred due to selfishness of agents, {\em and} lack of
coordination among them. The PoS on the other hand measures social cost losses
only due to selfishness: if all players coordinate (even by following some
externally provided instructions, or by interacting on the basis of a common
protocol) they may reach the least expensive equilibrium. The notion of strong
equilibria inherently permits coordination among subsets of agents. Therefore
we can intuitively expect the SPoA to almost match the PoS. Our results
confirm this intuition.

The study of the price of stability in network design games with fair
allocation of network link costs was initiated in~\cite{Anshelevich04}. One
challenging problem remaining open since then is improving upon an upper bound
of $PoS=O(\log n)$ and a lower bound (recently shown in~\cite{Fiat06}) of
$\frac{12}{7}$, for unweighted players on undirected networks. A series of
recent works~\cite{Fiat06,Chekuri06,Chen06,Epstein07,Albers08} provides
results (polylogarithmic upper bounds) with respect to the social cost of pure
Nash and (approximate) strong equilibria in the network design game model
of~\cite{Anshelevich04}, also for the case of weighted players. We review
these results in section~\ref{section:related-work}. We also explain how the
Facility Location game that we study is a special case of the more general
model introduced in~\cite{Anshelevich04}, that also includes {\em ``delay''}
costs in using network links (refered to as connection costs in the case of
Facility Location), apart from fairly allocated {\em ``installation''} costs.
For the case of metric connection costs we were able to prove greatly improved
and almost tight bounds for the PoS and the SPoA; this makes the {\em metric}
Facility Location game and exceptional special case of the model introduced
in~\cite{Anshelevich04}.

\paragraph{Summary of Results}{
We analyze the PoS of the unweighted metric Facility Location game
(section~\ref{section:unweighted-metric-pos}), and prove constant upper and
lower bounds. Note that an $O(\ln n)$ general upper bound implied by the work
of~\cite{Anshelevich04} is tight for non-metric networks (we discuss this in
section~\ref{section:definitions}). Our technique relies on direct analysis of
social cost evolution during an iterative best response performed by the
agents, until they reach equilibrium. We show in particular that given any
initial configuration (including the social optimum), the social cost of the
reached equilibrium is at most $2.36$ times the initial social cost and at
least $1.45$ times this cost in the worst case. Metric Facility Location is
the first case of the model proposed in~\cite{Anshelevich04}, in which an
additional structural network property (triangle inequality) yields a constant
almost tight Price of Stability. It is also interesting that incorporation of
link delays (connection costs) makes the previously derived $O(H(n))$ upper
bounds tight for the non-metric unweighted game. Subsequently we study strong
equilibria (section~\ref{section:apx-strong-equilibria}). Although they do not
always exist, we prove that $\alpha$-approximate strong equilibria exist (no
subset deviation causes factor $\alpha$ improvement to all of its members),
for $\alpha\geq e=2.718\dots$ in general networks and weighted agents. For the
Price of Anarchy of approximate strong equilibria (SPoA) we show an $O(\ln W)$
general upper bound which becomes $\Theta(\ln n)$ for unweighted agents
(section~\ref{section:strong-poa}). For the metric unweighted case we prove a
constant upper bound (section~\ref{subsection:metric-unweighted-spoa}).

Except for the SPoA analysis in
paragraph~\ref{subsection:metric-unweighted-spoa}
(theorem~\ref{theorem:metric-spoa-ub}), the rest of the described results have
appeared in~\cite{Hansen08}, in a shorter version~\footnote{A result of
$2.36$-approximate strong equilibria that appeared in a previous version of
this report was erroneous and has been removed.}. We note that the analysis of
the SPoA of approximate strong equilibria for the metric unweighted case in
paragraph~\ref{subsection:metric-unweighted-spoa} essentially generalizes the
analysis of the corresponding result appearing in~\cite{Hansen08} (theorem 2);
this result concerned the SPoA of exact strong equilibria only {\em when} they
exist.
}

\begin{table}[t]
\center
\begin{tabular}{c|c|c|c}
 & Metric & \multicolumn{2}{c}{Non-metric}\\
\cline{3-4}
 & Unweighted & Unweighted & Weighted\\
\hline
PNE $|$ SE & $\surd$ $|$ $e$-apx & $\surd$ $|$ $e$-apx & $e$-apx $|$ $e$-apx\\
PoS & $\in(1.45,2.36)$ & $\Theta(\ln n)$ & $O(\ln W)$\\
SPoA & $O(1)$ & $\Theta(\ln n)$ & $O(\ln W)$\\
\hline
\end{tabular}
\caption{Summary: Facility Location Games with Fair Allocation of Facility Costs.}
\label{table:results}
\end{table}

\section{Related Work}
\label{section:related-work}

Anshelevich {\em et al.} first studied the Price of Stability for network
design games with general~\cite{Anshelevich03} and fair cost
allocation~\cite{Anshelevich03}. In the latter case $n$ agents wish to connect
node subsets over a network, by strategically selecting links to use. Every
network link is associated to two cost components, an {\em installation} cost
and a {\em delay} cost. Link installation cost is shared fairly among agents
using the same link in~\cite{Anshelevich04}, while every agent experiences a
delay cost given by a polynomial function of the number of agents using the link. The authors showed that {\em for unweighted agents} these network design games belong to the class of {\em potential} games introduced by Monderer and
Shapley in~\cite{Monderer}, hence they have pure strategy Nash equilibria. In
particular, a potential function defined in~\cite{Monderer} is associated to
these games in the following manner: an arbitrarily initialized iterative best
response procedure carried out by the players reaches a local minimum of the
potential function, which corresponds to a pure strategy Nash equilibrium of
the game. In~\cite{Anshelevich04} the authors developed elegant arguments
using the potential function, so as to upper bound the PoS for unweighted
agents. In case of polynomial delay costs of degree at most $k$ the PoS was
shown to be at most $O((k+1)\ln n)$. For the case of directed networks this
bound was shown to be tight (for k=0 and zero delays), whereas for undirected
networks a lower bound of $\frac{4}{3}$ was given. This was recently improved
to $\frac{12}{7}$ by Fiat et al.~\cite{Fiat06}, who also showed an upper bound
of $O(\log\log n)$ for $PoS$, when exactly one agent resides on each node of
the network and all agents need to connect to a common sink node (single-sink
case). For weighted agents a potential function proving existence of pure
equilibria was developed for only $2$ agents in~\cite{Anshelevich04}. Chen et
al.~\cite{Chen06} showed that pure strategy Nash equilibria do not always
exist for weighted network design games with fair allocation of link
installation costs. They developed a trade-off for the social cost of
approximate equilibria versus the approximation factor, as a function of
the maximum weight taken over all agents.

An extensive study of {\em strong} equilibria for network design games with
fair cost allocation (but without delays) appeared recently by Albers
in~\cite{Albers08}, for unweighted and weighted agents. Strong equilibria do
not always exist, but she showed that $O(H(n))$ and $O(\ln W)$-approximate
strong equilibria do exist ($W$ is the sum of the agents' weights) for
unweighted and weighted games respectively. She proved $O(H(n))$ and $O(\ln
W)$ upper bounds for the SPoA of approximate strong equilibria which are tight
for directed graphs. For undirected graphs she showed $\Omega(\sqrt{\log n})$
and $\Omega(\sqrt{\log W})$ lower bounds. Finally and most importantly, she
showed an $\Omega(\frac{\log W}{\log\log W})$ lower bound for the PoS of
weighted network design games. Strong equilibria in the context of
(single-sink) unweighted network design games with fair cost sharing were
first studied in~\cite{Epstein07}. The authors gave topological
characterizations for the existence of strong equilibria on directed networks, and proved that $SPoA=\Theta(\log n)$.

As of the recent literature, the gap for the identification of the PoS of
network design games with fair cost allocation remains open for unweighted
agents, with the lower bound of $\frac{12}{7}$ being the best known so far, to
the best of our knowledge. In effect, this means that the impact of an
undirected network structure to the PoS is not well understood, since the
upper bounding potential function arguments developed in~\cite{Anshelevich04}
do not incorporate network structure. The Facility Location game is a special
case of the model studied in~\cite{Anshelevich04}, that is interesting on its
own right: it finds numerous applications and exhibits intriguing
characteristics. It emboddies non-shareable delays explicitly and in a sense
specializes {\em single-sink} network design considered
in~\cite{Epstein07,Chen06}: we can simply augment the network with a node $t$
and set links $(v,t)$ to have fairly shareable cost, equal to the facility
opening cost at $v$. The original network links have a delay cost only. Then
every agent needs to choose at most two edges from the node it resides on, to
$t$.

\paragraph{Facility Location Games}{
An unweighted metric facility location game with uniform facility costs was
studied in the context of {\em selfish caching} in~\cite{Chun04}. The network
model of that work is essentially equivalent to the one we discuss here, apart
from the fact that fair cost-sharing of facility costs was {\em not} used: an
agent could connect to a facility payed exclusively by another agent. Another
difference is that facility opening cost was uniform accross all nodes of the
network, whereas we consider node-dependent costs. The PoA and PoS of this
model were shown to be unbounded. The authors devised an extension of their
game model, with payments exchanged among agents, in which the socially
optimum configuration of the original game is rendered a pure strategy Nash
equilibrium (hence $PoS=1$ for the extended game). Vetta studied a class of
games for {\em competitive} facility location~\cite{Vetta06}, for which he
proved existence of pure strategy Nash equilibria and an upper bound of $2$
for the PoA. In this competitive setting enterprises open facilities at
certain nodes and try to attract customers to connect to them. He also
illustrated applicability of his model in the context of the $k$-median
problem model. In~\cite{Chekuri06} the authors considered {\em single-sink}
network design with fair cost-sharing of all resources (network links) used by
agents. They argued how this model can be viewed as a Facility Location model,
but did not take connection costs (non-shareable delays) into account.
}

\section{Definitions and Preliminaries}
\label{section:definitions}

The network will be a complete graph $G(V, E)$, each edge $(u,v)$ of which is
associated to a non-negative cost $d(u,v)$. We consider a set $A$ of $n$
agents. Each agent $i$ resides on a node $u_i\in V$, and is associated to a
non-negative demand weight $w_i$. The strategy space of agent $i$ is $V$: $i$
chooses a location $v\in V$, where it will receive service from. Opening a
facility at node $v\in V$ costs $\beta_v$. A strategy profile (configuration)
is denoted by $s=(s_1,\dots, s_n)$, where $s_i\in V$. The total weight of
agents receiving service from $v$ under $s$ is $W_s(v)$. The cost $c_i(s)$
experienced by agent $i$ in $s$ is:

\[
c_i(s)=w_id(u_i,s_i)+\frac{w_i\beta_{s_i}}{W_s(s_i)}
\]

\noindent Every agent $i$ pays a fraction $\frac{w_i}{W_s(s_i)}$ of the
facility installation cost at $s_i$. Denote by $F_s\subseteq V$ the set of
facility locations specified under $s$. The {\em social cost} $c(s)$ is then:

\[
c(s)=
\sum_ic_i(s)=
\sum_iw_id(u_i,s_i)+\sum_i\frac{w_i\beta_{s_i}}{W_s(s_i)}=
\sum_iw_id(u_i,s_i)+\sum_{v\in F_s}\beta_v
\]

\noindent We use $W(I)$ for the sum of weights of agents in set $I$. By $c_I(s)$ we denote the total cost of agents in $I$ under $s$. Furthermore, given a facility $v\in F_s$, we use $c_v(s)$ to denote $\sum_{i:s_i=v}c_i(s)$.

\medskip

\noindent \textbf{Pure Nash Equilibria} The {\em unweighted} Facility Location
game specializes network design games first studied in~\cite{Anshelevich04},
and is a potential game~\cite{Monderer}; it is associated to a potential
function, that can be locally minimized by iterative best response performed
by players, and the reached local minimum corresponds to a pure Nash
equilibrium. In iterative best response we choose iteratively an arbitrary player $i$, and let him/her decide a strategy that minimizes its individual cost with respect to the current configuration $s_{-i}=(s_1,\dots, s_{i-1},s_{i+1},\dots, s_n)$ for the rest of the players. The process ends when no player $i$ can improve his/her individual cost under $s_{-i}$. Then $s$ is a pure Nash equilibrium.

A simple example shows that the PoA of pure equilibria can be $n$ in the worst
case. Take a network of two nodes $u,v$. Opening a facility on either node has
a cost $\beta$. We set $d(u,v)=\frac{(n-1)\beta}{n}$, and let $n$ agents
reside on $u$. Then if every agent $i$ plays $s_i=v$, $i=1\dots n$, the
resulting configuration is a pure Nash equilibrium of cost $n\beta$, because
every agent pays exactly $\beta$ and has no incentive to open a facility at
node $u$ by paying the same cost $\beta$. The social optimum consists of a
single facility opened on $u$ and all agents being locally serviced by it, at
zero connection cost. Thus the socially optimum cost is $\beta$.

Potential function arguments developed in~\cite{Anshelevich04} yield an upper
bound of $H(n)$ for the PoS, where $H(n)$ is the $n$-th harmonic number. This
bound is tight for unweighted non-metric Facility Location.
Fig.~\ref{figure:general-pos-lb} presents a non-metric lower bounding example
for unweighted agents. Assume uniform facility costs equal to 1. $2n$
unweighted agents reside on boldly drawn nodes, $n$ of them on $v_2$.
The social optimum consists of facilities on $v_1$ and $v_2$, where $v_2$
serves only the $n$ agents residing on $v_2$. The rest are served by $v_1$.
However they deviate to $v_2$ from $v_1$ one by one. It is then: $PoS\geq
(2H(n)-H(2n))/(2+n\epsilon)\geq (\ln n-\ln 2-1)/(2+n\epsilon)=\Omega(\ln n)$.
When it comes to metric networks however, the PoS is constant as we show in
the following section. This result essentially indicates the impact of the
network's structure (triangle inequality) to the social cost of efficient pure
Nash equilibria.

\begin{figure}[t]
\center
\includegraphics[scale=0.5]{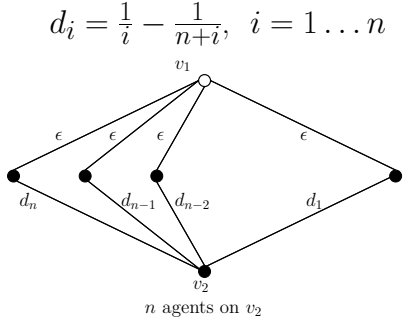}
\caption{Lower bounding example for unweighted (non-metric) PoS.}
\label{figure:general-pos-lb}
\end{figure}

\begin{definition}{\bf (Strong Equilibria)~\cite{Andelman07,Aumann59}}
A strategy profile $s$ is a strong equilibrium if no subset of agents can
deviate in coordination, so that each and every one of its members is better
off. It is an $\alpha$-approximate strong equilibrium if no subset of agents
can deviate in coordination, so that each and every one of its members is
better off by a factor strictly more than $\alpha\geq 1$.
\end{definition}

\section{Unweighted Metric Facility Location: The Price of Stability}
\label{section:unweighted-metric-pos}

We analyze evolution of an equilibrium through iterative best response, initialized at configuration $s^*$. In fact, for {\em any} strategy profile $s$ that is reached through iterative best response initialized at $s^*$, we will upper bound the ratio $c(s)/c(s^*)$. Thus we do not require that $s^*$ be the socially optimum configuration, nor that $s$ is an equilibrium strategy profile. We only require that $s$ is reached by iterative best response initialized at $s^*$. Then taking $s^*$ to be the socially optimum configuration and $s$ to be the reached equilibrium will yield an upper bound for the Price of Stability, as a corollary. To facilitate clarity, we give a brief description of the analysis' plan. Given any facility node $v\in F_{s^*}$ denote by $A_{s^*}(v)=\{i|s_i^*=v\}$ the subset of agents connected to $v$ in $s^*$. Then $c_v(s^*)$ will be the total cost incurred by $A_{s^*}(v)$ collectively. We will upper bound $c(s)/c(s^*)$ as follows:

\begin{equation}
\frac{c(s)}{c(s^*)}\leq
\frac{
\sum_{v\in F_{s^*}}\sum_{i\in A_{s^*}(v)}c_i(s)
}{
\sum_{v\in F_{s^*}}\sum_{i\in A_{s^*}(v)}c_i(s^*)
}=
\frac{
\sum_{v\in F_{s^*}}\sum_{i\in A_{s^*}(v)}c_i(s)
}{
\sum_{v\in F_{s^*}}c_v(s^*)
}\leq
\max_{v\in F_{s^*}}\frac{
\sum_{i\in A_{s^*}(v)}c_i(s)
}{
c_v(s^*)
}
\label{equation:plan}
\end{equation}

\noindent Define $\Delta c(i)=c_i(s)-c_i(s^*)$ to be the increase in cost caused by agent $i$ during iterative best response initialized at  $s^*$ and reaching $s$. For any subset $A'\subseteq A$ of agents we also use $\Delta c(A')=\sum_{i\in A'}(c_i(s)-c_i(s^*))$. To evaluate the upper bounding expression~(\ref{equation:plan}), we are going to use an upper bound on $\Delta c (A_{s^*}(v))=\sum_{i:s_i^*=v}\Delta c(i)$, valid {\em for any} $v\in F_{s^*}$. This is the increase in social cost caused during iterative best response by agents connected to $v$ in $s^*$. Then for any $v\in F_{s^*}$ we will maximize $\frac{c_v(s^*)+\Delta c(A_{s^*}(v))}{c_v(s^*)}$, by lower bounding $c_v(s^*)$ for any $v\in F_{s^*}$. In determining an upper bound on $\Delta c(A_{s^*}(v))$ we prove the following lemma, which charges any specific agent $i$ a bounded amount of social cost increase during the execution of iterative best response.

\begin{lemma}
Let $A_{s^*}(v)$ be the subset of agents that are connected to $v$ in $s^*$. For any $i\in A_{s^*}(v)$ define $A_{s^*}^i(v)\subseteq A_{s^*}(v)$ to be the subset of agents that have not yet deviated from $v$, exactly before the first deviation of $i$. Then $\Delta c(i)\leq \beta_v/|A_{s^*}^i(v)|$.
\label{lemma:ibr-charging}
\end{lemma}

\begin{proof}
For simplicity let $|A_{s^*}^i(v)|=k_i(v)$. Clearly $i\in A_{s^*}^i(v)$. Let us analyze contribution of $i$ to social cost increase during its first deviation. By deviating, $i$ reduces its individual cost from $c_i$ to $c_i'$, by joining another facility node $v'$. Then:

\[
c_i=x_i(v)+\frac{\beta_v}{k_i(v)}
\mbox{,\ \ \ }
c_i'=x_i(v')+\frac{\beta_{v'}}{\lambda_i(v')}
\]

\noindent $x_i(v)$ and $x_i(v')$ is the connection cost payed by $i$ before and after its first deviation. $\lambda_i(v')$ is the number of agents sharing facility cost at $v'$, including $i$. Since $c_i(v')<c_i(v)$, it is  $x_i(v')-x_i(v)\leq \frac{\beta_v}{k_i(v)}-\frac{\beta_{v'}}{\lambda_i(v')}$. Let $\Delta c_v(i)=c_i'-c_i$ be the difference caused to the social cost by this single deviation of $i$. We examine the following cases:

\begin{enumerate}

\item $k_i(v)>1$, $\lambda_i(v') > 1$: Then 
$\Delta c_v(i) = x_i(v')-x_i(v)\leq \frac{\beta_v}{k_i(v)}-\frac{\beta_{v'}}{\lambda_i(v')}$.

\item $k_i(v)=1$, $\lambda_i(v') >1 $: Then 
$\Delta c_v(i) = -\beta_v + x_i(v') - x_i(v)\leq -\frac{\beta_{v'}}{\lambda_i(v')}$.

\item $k_i(v)>1$, $\lambda_i(v') = 1$: Then $\Delta c_v(i) = \beta_{v'}+x_i(v')-x_i(v)\leq \frac{\beta_v}{k_i(v)}$.

\item $k_i(v)=1$, $\lambda_i(v')=1$: Then $\Delta c_v(i) = \beta_{v'}-\beta_v+x_i(v')-x_i(v)\leq 0$.

\end{enumerate}

\noindent Clearly the above hold in general during execution of iterative best response, for any agent that performs a single deviation from a node $v$ to a node $v'$. 

Now we implement a charging procedure along with iterative best response. Give all agents an initial label $\ell(i)=i$, before executing iterative best response. The labels will be updated during iterative best response, so that  increases caused by agent $i$ will be charged to an appropriate agent $\ell(i)$. Initialize $\Delta c({\ell(i)})=0$ for all $i\in A$. We will make use of the auxiliary variables $k_i(v)$ and $\lambda_i(v')$, as they were described before. At any time $\lambda_i(v)$ is the number of agents (including $i$) connected to $v$, right after $i$ {\em joined} $v$. We need to initialize $\lambda_i(v)$ to a value for every $v\in F_{s^*}$ and $i\in A_{s^*}(v)$, to implement the charging scheme. Set $\lambda_{\ell(i)}=\lambda_i$ to a distinct value from $\{1,2,\dots, |A_{s^*}(v)|\}$. $k_{\ell(i)}(v)$ will be always (during iterative best response) the number of agents connected to $v$ exactly before deviation of $i$ from $v$ (before $i$ {\em leaves} $v$). Charging is then implemented by relabeling deviating agents in the following manner. 

\begin{enumerate}
	
\item If $k_{\ell(i)}(v)=\lambda_{\ell(i)}(v)$ no relabeling is needed. 

\item Otherwise there must be some $j\neq i$ connected to $v$ such that  $\lambda_{\ell(j)}(v)=k_{\ell(i)}(v)$. In this case exchange labels of $i$ and $j$. 

\end{enumerate}

\noindent Subsequently add the increase caused by deviation of $i$ to $\Delta c(\ell(i))$. Finally, set $\lambda_{\ell(i)}(v')$ equal to the number of agents connected to $v'$ right after $i$ has joined $v'$.

By the previous definitions it follows that if $k_{\ell(i)}(v)\neq\lambda_{\ell(i)}(v)$, then it is always $k_{\ell(i)}(v) >\lambda_{\ell(i)}(v)$, i.e. $i$ has joined $v$ before some agent $j$ with $\lambda_{\ell(j)}(v)=k_{\ell(i)}(v)$, but leaves $v$ ``out of order'', i.e. before $j$ leaves. By exchanging labels of $i$,$j$ we add the increase caused by $i$ to the agent that previously labeled $j$. Possible increases in 1.,2.,3.,4., imply that any agent is charged by the end of iterative best response at most $\frac{\beta_v}{|A_{s^*}^i(v)|}$ for some $i$. If we ``guess'' the exact order of first deviation of all agents, we can initialize $\lambda_i(v)=k_i(v)$. Then each $i$ itself will be charged at most $\frac{\beta_v}{|A_{s^*}^i(v)|}$.
\qed
\end{proof}

\bigskip

In what follows we are going to upper bound the ratio $\frac{\sum_{i\in A_{s*}(v)}c_i(s)}{c_v(s^*)}$ for any $v\in F_{s^*}$. 
If no agent $i\in A_{s^*}(v)$ ever deviates from playing $v$, then clearly it will be $\frac{\sum_{i\in A_{s*}(v)}c_i(s)}{c_v(s^*)}=1$. Our analysis focuses on two remaining cases: either {\em (i)} there is a non-empty set of agents $A_s(v)\subset A_{s^*}(v)$ that never deviate from $v$ during iterative best response or {\em (ii)} all agents $i\in A_{s^*}(v)$ deviate from $v$ during iterative best response. Case {\em (i)} is examined in proposition~\ref{intermmediate-case}, whereas the analysis is concluded by   analysis of {\em (ii)} in theorem~\ref{theorem:metric-pos-ub}.

\begin{proposition}
Let $s$ be a strategy profile reached by iterative best response initialized at strategy profile $s^*$. For any facility $v\in F_{s^*}$ define $A_{s^*}(v)=\{i|s_i^*=v\}$ and let $A_{s}(v)\subset A_{s^*}(v)$ the subset of agents that never deviated from $v$, during iterative best response. If $A_s(v)\neq \emptyset$, then $\sum_{i\in A_{s^*}(v)}c_i(s)\leq 2.36\cdot c_v(s^*)$.
\label{intermmediate-case}
\end{proposition}

\begin{proof}
As discussed previously, we will upper bound the ratio $\frac{c_v(s^*)+\Delta c(A_{s^*}(v))}{c_v(s^*)}$, by deriving an upper bound for $\Delta c(A_{s^*}(v))$ and a lower bound for $c_v(s^*)$. By lemma~\ref{lemma:ibr-charging} follows immediately:

\begin{equation}
\Delta c(A_{s^*}(v))\leq \beta_v[H(|A_{s^*}(v))-H(|A_s(v)|)]
\label{equation:delta-ub}
\end{equation}

\noindent because agents in $A_s(v)$ never deviated, whereas the rest caused a cost increase equal to $\beta_v$ times a harmonic series term each. In what follows we show a lower bound on $c_v(s^*)$.

Fix any facility $v\in F_{s^*}$, and define $I_{s^*}(v)=A_{s^*}(v)\setminus A_s(v)$. Without loss of generality assume such an order of agents, that agents $i\in I_{s^*}(v)$  (and for every $v\in F_{s^*}$) ``best-respond'' consecutively; e.g. assume that the algorithm scans the facilities in $F_{s^*}$ in an arbitrary fixed order and for each facility, the agents connected to it in an arbitrary fixed order. We focus only on the first deviation of each agent $i\in I_{s^*}(v)$ for some $v\in F_{s^*}$. Let $v'$ be the node that $i$ deviates to. For convenience we  use  $x_i^*=d(u_i,v)$ and let $\delta x_i^*=d(u_i, v')-d(u_i,v)$. Let $\lambda_i$ be the total number of agents connected to $v'$ {\em right after} deviation of $i$. The new cost of $i$ right after its first deviation is: $x_i^*+\delta x_i^*+\frac{\beta_{v'}}{\lambda_i}$. For any other agent $j\in A_{s^*}(v)\setminus\{i\}$ that, either deviates from $v$ to some node $v''$ after the deviation of $i$, or never deviated from $v$ (in this case consider a trivial deviation with $v''=v$), we have respectively:

\begin{equation}
d(u_j, v)+\delta x_j^*+\frac{\beta_{v''}}{\lambda_j}\leq
d(u_j, v')+\frac{\beta_{v'}}{\lambda_i}
\label{equation:deviation-preference}
\end{equation}

\noindent Sub\-sti\-tu\-te $d(u_j, v')$ in~(\ref{equation:deviation-preference}) by tri\-an\-gle in\-eq\-ua\-li\-ty: $d(u_j, v')\leq d(u_j, v)+d(u_i, v)+d(u_i, v')$. Let $k_i^*$ denote the number of agents connected to $v$ {\em right before} deviation of $i$ from $v$. Also  note that $\delta x_i^*+\frac{\beta_{v'}}{\lambda_i}\leq\frac{\beta_v}{k_i^*}$, because $i$ decreases its individual cost by deviating. Thus we obtain:

\begin{equation}
d(u_i, v)
\geq
\frac{1}{2}\Bigl(\delta x_j^*-\delta x_i^*+\frac{\beta_{v''}}{\lambda_j}-\frac{\beta_{v'}}{\lambda_i}\Bigr)
\geq
\frac{1}{2}\Bigl(\delta x_j^*+\frac{\beta_{v''}}{\lambda_j}-\frac{\beta_{v}}{k_i^*}\Bigr)
\label{equation:connection-lb-0}
\end{equation}

\noindent Because $d(u_i, v)\geq 0$, and because the latter has to hold for every pair of distinct agents $i,j\in A_{s^*}(v)$, we deduce:

\begin{equation}
d(u_i, v)
\geq \frac{1}{2}\left[
\max_{j:s_j^*=v}\left(
\delta x_j^*+\frac{\beta_{v''}}{\lambda_j}
\right)
-\frac{\beta_v}{k_i^*}
\right]
\label{equation:connection-lb-1}
\end{equation}

\noindent We use~(\ref{equation:connection-lb-1}) for the connection cost of agents in $I_{s^*}(v)$, and $0$ for agents in $A_s(v)=A_{s^*}(v)\setminus I_{s^*}(v)$; in essence for every $i\in A_s(v)$ we simply set in~(\ref{equation:connection-lb-1}) $v''=v$, $k_i^*=|A_s(v)|$, and take $\max_{j:s_j^*=v}\left(\delta x_j^*+\frac{\beta_{v''}}{\lambda_j}\right)=\frac{\beta_{v}}{|A_s(v)|}$. The cost $c_v(s^*)$ is then:

\begin{equationarray}{lcl}
c_v(s^*) & \geq &
\beta_v +
\frac{\beta_v}{2}
\sum_{k=|A_s(v)|}^{|A_{s^*}(v)|}
\Bigl(\frac{1}{|A_s(v)|}-\frac{1}{k}\Bigr)\nonumber\medskip\\
& = &
\beta_v+
\frac{\beta_v}{2}
\Bigl(\frac{|A_{s^*}(v)|-|A_s(v)|}{|A_s(v)|}-H(|A_{s^*}(v)|)+H(|A_s(v)|)\Bigr)
\label{equation:v-cost-lb}
\end{equationarray}

\noindent Using the bounds~(\ref{equation:v-cost-lb}) and~(\ref{equation:delta-ub}) we deduce the following upper bound:

\begin{equationarray}{lcl}
\left(\sum_{i\in A_{s^*}(v)}c_i(s)\right)/c_v(s^*) & \leq &
\frac{
1+\frac{1}{2}\Bigl(
\frac{|A_{s^*}(v)|-|A_s(v)|}{|A_s(v)|}+H(|A_{s^*}(v)|)-H(|A_s(v)|)
\Bigr)
}{
1+\frac{1}{2}\Bigl(
\frac{|A_{s^*}(v)|-|A_s(v)|}{|A_s(v)|}-H(|A_{s^*}(v)|)+H(|A_s(v)|)
\Bigr)
}\nonumber\medskip\\
& \leq &
\frac{
1.5+|A_{s^*}(v)|/|A_s(v)|+\ln(|A_{s^*}(v)|/|A_s(v)|)
}{
0.5+|A_{s^*}(v)|/|A_s(v)|-\ln(|A_{s^*}(v)|/|A_s(v)|)
}\label{equation:metric-pos-ub}
\end{equationarray}

\noindent The latter inequality follows by bounds for the harmonic number, i.e. $\gamma+\ln m \leq H(m)\leq 1+\ln m$, where $\gamma > 0.5$ is the {\em Euler} constant. By setting $y=|A_{s^*}(v)|/|A_s(v)|$, we can numerically maximize the upper bound with respect to $y$, to at most $2.36$, for $y\simeq 2.47$.\qed
\end{proof}

\begin{theorem}
If $s$ is a strategy profile for the unweighted Metric Facility Location game, reached by iterative best response initialized at a strategy profile $s^*$, then $c(s^*)\leq 2.36 c(s)$.
\label{theorem:metric-pos-ub}
\end{theorem}

\begin{proof}
We only need to complement the result of proposition~\ref{intermmediate-case} by considering the case where $A_s(v)=\emptyset$. That is, $I_{s^*}(v)=A_{s^*}(v)$, and all agents connected to $v$ in $s^*$ deviate during iterative best response. If for each $i\in I_{s^*}(v)$ that deviates from $v$ to $v'$ ($v'$ is not necessarily the same for every $i$) it is $\delta x_i^*+\frac{\beta_{v'}}{\lambda_i}\leq \frac{\beta_v}{|A_{s^*}(v)|}$, then:

\[
\sum_{i\in A_{s^*}(v)}c_i(s)=
\sum_{i\in A_{s^*}(v)}\left(
x_j^*+\delta x_i^*+\frac{\beta_{v'}}{\lambda_i}
\right)\leq \beta_v+\sum_{i\in A_{s^*}(v)}x_i^*=c_v(s^*)
\]

\noindent Thus, assume there is at least one agent $j$ with $\delta x_j^*+\frac{\beta_{v'}}{\lambda_j} > \frac{\beta_v}{|A_{s^*}(v)|}$. In general, let $r$ be the {\em largest} integer, so that $\delta x_i^*+\frac{\beta_{v'}}{\lambda_i}\leq \frac{\beta_v}{r}$ for all $i\in I_{s^*}(v)=A_{s^*}(v)$. Using the same arguments that led to~(\ref{equation:connection-lb-1}), and taking $\max_i\left(x_i^*+\frac{\beta_{v'}}{\lambda_i}\right)=\frac{\beta_v}{r}$, we deduce:

\begin{equation}
x_i^*=d(u_i, v)\geq
\frac{1}{2}\left(\frac{\beta_v}{r}-\frac{\beta_v}{k_i^*}\right)
\mbox{,\ \ for $k_i^*\geq r$\ \ \ and\ \ \ $x_i^*=d(u_i,v)\geq 0$,\ \ for $k_i^*\leq r$}
\label{equation:x-lb}
\end{equation}

\noindent Using~(\ref{equation:x-lb}) we end up with a similar lower bound to~(\ref{equation:v-cost-lb}) for $1\leq r\leq |A_{s^*}(v)|$:

\begin{equation}
c_v(s^*) \geq
\beta_v +
\frac{\beta_v}{2}
\sum_{k=r}^{|A_{s^*}(v)|}
\Bigl(\frac{1}{r}-\frac{1}{k}\Bigr)
=
\beta_v+
\frac{\beta_v}{2}
\Bigl(\frac{|A_{s^*}(v)|-r}{r}-H(|A_{s^*}(v)|)+H(r)\Bigr)
\label{equation:v-cost-lb-1}
\end{equation}

\noindent To finish the proof, we note that by~(\ref{equation:x-lb}), deviation of agents with $k_i^*\leq r$, yields a total cost increase (payed as new connection cost) of at most $r\times\frac{\beta}{r}=\beta_v$, which is exactly equal to the cost saved by closing the facility at $v$ (because all agents deviate). Thus we can simulate the situation with the case analyzed in proposition~\ref{intermmediate-case}, by setting $A_s(v)=\{i\in A_{s^*}(v)|k_i^*\leq r\}$ and $r=|A_s(v)|$, so as to end up with~(\ref{equation:metric-pos-ub}). This is technically equivalent to assuming that agents $i$ with $k_i^*\leq r$ never deviated.\qed
\end{proof}

\begin{corollary}
The Price of Stability for the unweighted Metric Facility Location game with fairly allocated facility costs is at most $2.36$.
\end{corollary}

\subsection{A Lower Bound on the Price of Stability}

We analyze a worst-case example, that makes the PoS for unweighted metric Facility Location at least $1.45$ asymptotically. Experimental evidence shows a lower bound $>1.77$, which we believe is tight for PoS. This is because the analysis of theorem~\ref{theorem:metric-pos-ub} embodies some losses due to bounding of harmonic numbers by logarithms.

Our construction appears in Fig.~\ref{figure:metric-pos-lb}(a). Take $2n$ agents, $n$ of them residing on a singe node $v$ of the network. Facility opening costs are assumed $1$ everywhere. The social optimum $s^*$ has $1+\sqrt{n}$ facilities: $v$ and $v_l^*$, $l=1\dots k$,  where $k=\sqrt{n}$. Let the $n$ agents residing on $v$ be serviced by $v$ in $s^*$, whereas the rest are equally partitioned to facilities $v_l^*$, $l=1\dots k$; every facility $v_l^*$ services a batch of $k=\sqrt{n}$ agents. In the least expensive equilibrium $s$, all agents from each facility $v^*_l$ of $s^*$ will deviate and miss-connect to $v$. We will analyze for {\em any} single facility $v_l^*$ - henceforth denoted by $v^*$ - the ratio $\left(\sum_{i\in A_{s^*}(v^*)}c_i(s)\right)/c_{v^*}(s^*)$. Analysis for the rest of the facilities will be similar, and will yield asymptotically the same result. In particular, by choosing to analyze {\em any} batch of $k$ agents that deviate from $v^*$ to $v$, the PoS will be lower bounded as:

\begin{equation}
PoS=
\lim_{n\rightarrow\infty}
\frac{
1+\sum_{l=1}^k\sum_{i\in A_{s^*}(v^*_l)}c_i(s)}{1+\sum_{l=1}^kc_{v^*_l}(s^*)}
\geq
\lim_{n\rightarrow\infty}\frac{
\sum_{i\in A_{s^*}(v^*)}c_i(s)
}{
(1/\sqrt{n}+c_{v^*}(s^*))
}
\label{equation:metric-pos-lb}
\end{equation}

For some {\em constant} $p\in (0,1)$ we will have $r=\lceil(1-p)k\rceil$ of the agents in $A_{s^*}(v)$ increase their connection cost significantly by deviating, hence the social cost. For every $i\in A_{s^*}(v)$ let $x_i^*$ be the distance of $i$ from $v^*$, and $\delta x_i^*$ the increase in connection cost after deviation. We determine $x_i^*$ and $\delta x_i^*$ by following  best responses of agents in $A_{s^*}(v^*)$. Let $\lambda\geq n$ be the number of agents already connected to $v$ before agents of $A_{s^*}(v^*)$ deviate to $v$ in order. In particular:

\begin{itemize}
	
\item  For $i=1\dots r$ set $\delta x_i^*=\frac{1}{k-i+1}-\frac{1}{\lambda+i}-\epsilon$; a small decrease $\epsilon>0$ causes $i$ to deviate to $v$.

\item For each of the remaining $k-r$ agents set $\delta x_i^*=\frac{1}{k-r+1}$.

\end{itemize}

By similar arguments as for the derivation of (\ref{equation:deviation-preference}),(\ref{equation:connection-lb-0}) and triangle inequality, we obtain $x_i^*=d(u_i,v^*)=\max\{0,\frac{1}{2}(\max\delta x_j^*-\delta x_i^*)\}$, where $\max_j\delta x_j^*=\frac{1}{k-r+1}$. This yields $x_i^*=0$ for $i>r$ (the remaining $k-r$ agents). Note that the $k-r$ remaining agents will deviate to $v$, because they prefer to pay strictly less than $\frac{1}{k-r+1}+\frac{1}{\lambda+r}$, instead of a cost share $\frac{1}{k-r}$ each, for the facility at $v^*$. By deviation they ``shut-down'' the facility at $v^*$ and decrease the social cost by $1$, but pay a total connection cost $\frac{k-r}{k-r+1}$ for $v$, which tends to $1$ for large $n$. Let us now derive the cost $c_{v^*}(s^*)$. For convenience we use the notation $\Delta H(p,q)$ to denote $H(p)-H(q)$. For the first $r$ agents, summing up as in~(\ref{equation:v-cost-lb}) yields:

\begin{equation}
c_{v^*}(s^*)=
1+\frac{1}{2}
\Bigl[\frac{r}{k-r+1}- \Delta H(k,k-r)\Bigr]-
\frac{1}{2}\frac{r}{\lambda+r}+\frac{1}{2}\Delta H(\lambda+r,\lambda)
\label{equation:v-opt}
\end{equation}

\noindent The cost of agents in $A_{s^*}(v^*)$ after deviation is: $\sum_{i\in A_{s^*}(v)}c_i(s)= c_{v^*}(s^*)-1+\frac{k-r}{k-r+1}+\sum_{i=1}^r\delta x_i^*$. By substituting $\sum_{i=1}^r\delta x_i^*=\Delta H(k, k-r)-\Delta H(\lambda+r,\lambda)$, we obtain:

\begin{figure}[t]
\center
\subfloat[Lower bounding example.]{
\begin{minipage}{120pt}
\center
\includegraphics[scale=0.5]{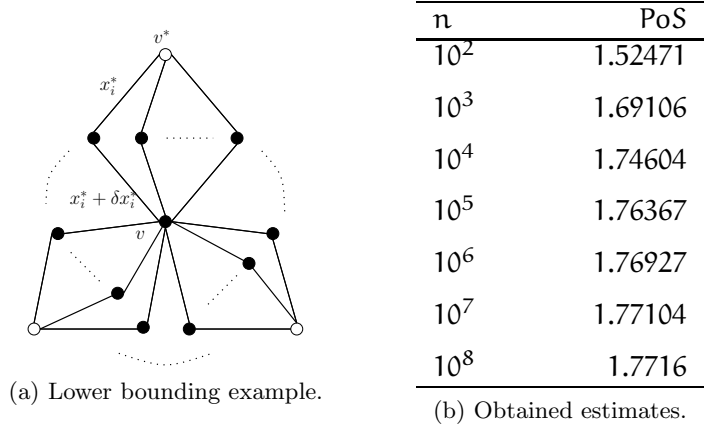}
\end{minipage}
}\quad\quad
\subfloat[Obtained estimates.]{
\begin{minipage}{120pt}
\center
\begin{tabular}{lr}
\hline
$n$\quad & $PoS$\\
\hline
$10^2$\mbox{\ \ \ \ } &\mbox{\ \ \ \ } $1.52471$\medskip\\
$10^3$\mbox{\ \ \ \ } &\mbox{\ \ \ \ } $1.69106$\medskip\\
$10^4$\mbox{\ \ \ \ } &\mbox{\ \ \ \ } $1.74604$\medskip\\
$10^5$\mbox{\ \ \ \ } &\mbox{\ \ \ \ } $1.76367$\medskip\\
$10^6$\mbox{\ \ \ \ } &\mbox{\ \ \ \ } $1.76927$\medskip\\
$10^7$\mbox{\ \ \ \ } &\mbox{\ \ \ \ } $1.77104$\medskip\\
$10^8$\mbox{\ \ \ \ } &\mbox{\ \ \ \ } $1.7716$\\
\hline
\end{tabular}
\end{minipage}
}
\caption{Experimental estimation of a PoS lower bound for the unweighted metric case.}
\label{figure:metric-pos-lb}
\end{figure}

\begin{equation}
\sum_{i\in A_{s^*}(v)}c_i(s)=
\frac{k-r}{k-r+1}+\frac{1}{2}
\Bigl[\frac{r}{k-r+1}+\Delta H(k,k-r)\Bigr]-
\frac{1}{2}\left(
\frac{r}{\lambda+r}+\Delta H(\lambda+r,\lambda)
\right)
\label{equation:eq-cost}
\end{equation} 

\noindent We simplify~(\ref{equation:v-opt}) and~(\ref{equation:eq-cost}) and substitute to~(\ref{equation:metric-pos-lb}) appropriately. For $\Delta H(k,k-r)$ we use the following bounds:

\[
\gamma-1+\ln\frac{1}{p+1/k}\leq 
\Delta H(k,k-r)\leq
1-\gamma+\ln\frac{1}{p-1/k}
\]

\noindent where $\gamma>0.5$ is the {\em Euler} constant. Also we use $\frac{1-p-1/k}{p+2/k}\leq \frac{r}{k-r+1}\leq \frac{1-p}{p}$, because $(1-p)k-1\leq r=\lceil (1-p)k\rceil\leq (1-p)k+1$. Then~(\ref{equation:metric-pos-lb}) becomes:

\begin{small}
\begin{equation}
PoS \geq \lim_{n\rightarrow\infty}
\frac{
\frac{k-r}{k-r+1}+\frac{1}{2}
\Bigl[\frac{r}{k-r+1}+\Delta H(k,k-r)\Bigr]-
\frac{1}{2}\frac{r}{\lambda+r}-\frac{1}{2}\Delta H(\lambda+r,\lambda)
}{
\frac{1}{\sqrt{n}}+
1+\frac{1}{2}
\Bigl[\frac{r}{k-r+1}-\Delta H(k,k-r)\Bigr]-
\frac{1}{2}\frac{r}{\lambda+r}+\frac{1}{2}\Delta H(\lambda+r,\lambda)
}\Rightarrow
\label{equation:pos-exp}
\end{equation}
\end{small}

\[
PoS\geq
\lim_{n\rightarrow\infty}
\frac{
\frac{p\sqrt{n}-1/k}{p\sqrt{n}+2/k}+\frac{1}{2}
\Bigl[\frac{1-p-1/k}{p+2/k}+\gamma-1+\ln\frac{1}{p+1/k}\Bigr]-
\frac{1}{2}\frac{r}{\lambda+r}-\frac{1}{2}\Delta H(\lambda+r,\lambda)
}{
\frac{1}{\sqrt{n}}+
1+\frac{1}{2}
\Bigl[\frac{1-p}{p}+1-\gamma-\ln\frac{1}{p-1/k}\Bigr]-
\frac{1}{2}\frac{r}{\lambda+r}+\frac{1}{2}\Delta H(\lambda+r,\lambda)
}
\]

\noindent Notice that for large $n$, $\frac{r}{\lambda+r}\rightarrow 0$, because $r=(\sqrt{n})$, and $\lambda\geq n$. Furthermore it is $\Delta H(\lambda+r,\lambda)\rightarrow 0$, because $\Delta H(\lambda+r,\lambda)\leq \frac{r}{\lambda}$, with $\lambda\geq n$ and $r=O(\sqrt{n})$. Assuming that $p$ is a constant, the limits of numerator and denominator exist, and we can use them to calculate the limit of the fraction. Given $\gamma > 0.5$ the PoS lower bound is simplified to:

\[
PoS > \left(1/4+\frac{1}{2}
\Bigl(\frac{1}{p}-\ln p\Bigr)\right)
/\left(3/4+\frac{1}{2}
\Bigl(\frac{1}{p}+\ln p\Bigr)
\right)
\]

Numerical maximization over $p\in (0,1)$ yields $PoS>1.45$, for $p\simeq 0.18$. We also searched computationally for $r$ maximizing~(\ref{equation:pos-exp}), for increasing $n$. For $r=0.27\sqrt{n}$ we found the values appearing in the table~\ref{figure:metric-pos-lb}(b) , that indicate $PoS > 1.77$. One can verify that any configuration other than $s^*$ and $s$ is more expensive, by definition of distances of agents from $v$ and $v_l^*$.

\section{Approximate Strong Equilibria}
\label{section:apx-strong-equilibria}

Strong equilibria do not generally exist in the Facility Location game, even for unweighted agents on metric networks with uniform facility costs. We illustrate this by an example, over the network shown in Fig.~\ref{figure:strong}. Consider the {\em metric} case in Figure~\ref{figure:strong}. There are three unweighted agents $i=1..3$ situated on distinct nodes $u_i$ of the depicted 6-node cycle. For any equilibrium strategy profile $s$ there exist two agents $i,j$, with $j = (i+1) \bmod 3$ with $c_i(s) \geq 1$, $c_j(s)\geq \frac{8}{9}$. By agreeing to open a facility on vertex $v_i$, $i$ and $j$ would change their costs to $c_i(s') = \frac{8}{9}$ and $c_j(s') = \frac{7}{9}$ respectively, each lowering their cost by at least $\frac{1}{9}$. Hence, this example does not have strong equilibria.

%%%
\begin{figure}[t]
\center
\includegraphics[width=5cm]{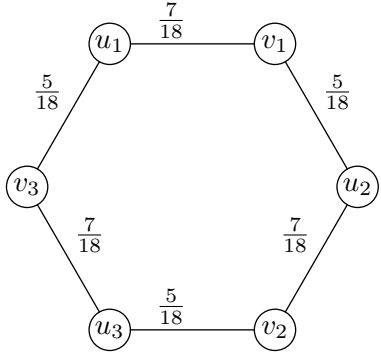}
\caption{Example of non-existence of strong equilibria.}
\label{figure:strong}
\end{figure}
%%%

Existence of pure equilibria for weighted agents is also an open issue. However, we were able to reduce to a constant  the logarithmic approximation factor $\alpha$ known for general network design~\cite{Albers08,Chen06}. In fact, our result is even more general, as it concerns {\em strong} equilibria. We make use of the following remark.

\begin{remark}
If an instance of the Facility Location game does not have strong equilibria, then there is at least one cycle of deviations of particular coalitions that results in a circular sequence of configurations $\{s^j\}_{j=1}^k$ with $s^1=s^k$. 
\label{cyclic-remark}
\end{remark}

Given such a sequence $\{s^j\}_{j=1}^k$, we denote the coalition that deviates from $s^j$ to form $s^{j+1}$ by $I_j$. Such a deviation causes a cost decrease of agents in $I_{j}$ and possibly a cost increase of agents in $A\setminus I_{j}$. Recall that $A$ is the set of all agents. We define two quantities, the {\em weighted improvement} $\impr(I_j)$ for agents in $I_{j}$ and the {\em weighted damage} $\dam(I_j)$ caused by agents in $I_j$ respectively:

\[
\impr(I_j) = \prod_{i \in I_{j}}\left(\frac{c_i(s^{j})}{c_i(s^{j+1})}\right)^{w_i}
\mbox{\ \ \ \ }
\dam(I_j) = \prod_{i \in A \setminus I_{j}}\left(\frac{c_i(s^{j+1})}{c_i(s^j)}\right)^{w_i}
\]

\noindent We derive an approximation factor that eliminates cycles.

\begin{lemma}
\label{lemma:damage}
Let $\{s^j\}_{j=1}^k$ with $s^1=s^k$ be a cycle of configurations in a Facility Location game instance, caused by consecutive deviations of coalitions. The game instance has an $\alpha$-approximate strong equilibrium if for all such sequences $\alpha \geq \dam_{\max}(\{s^j\}_{j=1}^k)$, where
$\dam_{\max}(\{s^j\}_{j=1}^k) = \displaystyle\max_{j=1\dots k-1}\dam(I_j)^{1/W(I_j)}$.
\end{lemma}
\begin{proof}
If there is no $\alpha$-approximate strong equilibrium we know that there is at least one cycle $\{s^j\}_{j=1}^k$ such that $\forall j \in \{1,\dots,k-1\} \forall i \in I_j: \frac{c_i(s^{j})}{c_i(s^{j+1})} > \alpha$.

\noindent Because $s^1=s^k$ we have that $\prod_{j=1}^{k-1} \frac{c_i(s^{j})}{c_i(s^{j+1})}=1$ for every agent $i$. Then:

\[
1 = \prod_{i=1}^n \left( \prod_{j=1}^{k-1} \frac{c_i(s^{j})}{c_i(s^{j+1})}\right)^{w_i} = 
\prod_{j=1}^{k-1} \frac{\impr(I_j)}{\dam(I_j)} >
\prod_{j=1}^{k-1} \frac{\alpha^{W(I_j)}}{\Bigl(\dam(I_j)^{1/W(I_j)}\Bigr)^{W(I_j)}} 
\]

\noindent It follows that $\dam_{\max}(\{s^j\}_{j=1}^k) > \alpha$. Hence the lemma follows by contradiction.\qed
\end{proof}

\bigskip

\noindent We derive an approximation factor as an upper bound of $\dam_{\max}(\{s^j\}_{j=1}^k)$ for any cycle.

\begin{theorem}
For every $\alpha\geq e$ there exist $\alpha$-approximate strong equilibria in the Facility Location game with fairly allocated facility costs, even for weighted agents and general networks.
\end{theorem}
\begin{proof}
We  prove that $\dam_{\max}(\{s^j\}_{j=1}^k) < e$ for every cycle $\{s^j\}_{j=1}^k$ of configurations and the result follows from Lemma~\ref{lemma:damage}. Let $I_j(v)$ be the set of agents going to $v$ in $s^j$, but not in $s^{j+1}$, and $A_j(v)$ be the set of agents going to $v$ in both $s^j$ and $s^{j+1}$. Note that $I_j = \bigcup_{v \in V} I_j(v)$, therefore:

\[
\dam_{\max}(\{s^j\}_{j=1}^k) =
\max_{j} \left(
\prod_{v \in V} \left(\prod_{i \in A_j(v)} \left(\frac{c_i(s^{j+1})}{c_i(s^j)}\right)^{w_i}
\right)^{\frac{W(I_j(v))}{W(I_j(v))}}\right)^{\frac{1}{W(I_j)}}\Rightarrow
\]

\[
\dam_{\max}(\{s^j\}_{j=1}^k) \leq
\max_{j,v} \left(\prod_{i \in A_j(v)}
  \left(\frac{c_i(s^{j+1})}{c_i(s^j)}\right)^{w_i}
  \right)^{\frac{1}{W(I_j(v))}}
\]

\noindent Hence, we need only consider what happens at the worst case node. For an agent $i$ in $A_j(v)$ we get that:

\[
\frac{c_i(s^{j+1})}{c_i(s^j)} = \frac{w_i \left(
  d(u_i,v) +
  \frac{\beta_v}{W_{s^{j+1}}(v)}\right)
  }{w_i \left(
  d(u_i,v) +
  \frac{\beta_v}{W(I_j(v)) + W(A_j(v))}\right)
  }
\leq
1+\frac{W(I_j(v))}{W(A_j(v))}
\]

\noindent It follows that:
\[
\dam_{\max}(\{s^j\}_{j=1}^k) \leq 
\max_{j,v} 
\left(
1 + \frac{W(I_j(v))}{W(A_j(v))}
\right)^{\frac{W(A_j(v))}{W(I_j(v))}}<
\lim_{r \to \infty} \left(1 + \frac{1}{r}\right)^r = e
\]
\qed
\end{proof}

\bigskip

\noindent Approximate strong equilibria are also approximate pure Nash equilibria, thus:

\begin{corollary}
The Facility Location game with fairly allocated facility costs and weighted agents has $\alpha$-approximate pure strategy Nash equilibria for every $\alpha\geq e$.
\end{corollary}

\section{The Strong Price of Anarchy}
\label{section:strong-poa}

\noindent We derive next an upper bound on the SPoA of $\alpha$-approximate strong equilibria, for the general Facility Location game.

\begin{theorem}
\label{theorem:spoa-ub}
For any constant $\alpha\geq e$, the Price of Anarchy of $\alpha$-approximate strong equilibria for the Facility Location game with fairly allocated facility costs, is upper bounded tightly by $O(H(n))$ for unweighted and by $O(\ln W)$ for weighted agents, where $W$ is the sum of weights.
\end{theorem}

\begin{proof}
Let $s$ and $s^*$ be the most expensive strong equilibrium and the socially optimum configuration respectively. For any facility node $v\in F_{s^*}$ let $A_{s^*}(v)$ and  $I_{s^*}(v)$ be respectively the subsets of agents connected to $v$ in $s^*$, and connected to $v$ in $s^*$ but not in $s$. We will upper bound the SPoA by $\max_{v\in F_{s^*}}\left[\left(\sum_{i\in A_{s^*}(v)}c_i(s)\right)/c_v(s^*)\right]$. Because $s$ is an $\alpha$-approximate strong equilibrium, {\em for every} subset of $I_{s^*}(v)$ there is at least one agent $i$, that is not willing to deviate to $v$ in coordination with the rest agents of the subset. Let  $I_{s^*}(v)=\{1,\dots,|I_{s^*}(v)|\}$ and $I_{s^*}^i(v)=\{1,\dots, i\}$, where $i$ is not willing to deviate in coordination with $I_{s^*}^i(v)\setminus\{i\}$. Then:

\[
c_i(s)\leq 
\alpha\Bigl(w_id(u_i, v)+\frac{w_i\beta_v}{W(I_{s^*}^i(v))+W_s(v)}\Bigr)
\leq
\alpha\Bigl(w_id(u_i, v)+\frac{w_i\beta_v}{W(I_{s^*}^i(v))}\Bigr)\Rightarrow\]

\[
\sum_{i\in I_{s^*}(v)}c_i(s)=\sum_{i\in I_{s^*}(v)}c_i(s)\leq 
\alpha\Bigl(\sum_{i\in I_{s^*}(v)}w_id(u_i,v)+
\beta_v\sum_{i=1}^{|I_{s^*}(v)|}\frac{w_i}{W(I_{s^*}^i(v))}\Bigr)
\Rightarrow
\]

\[
\sum_{i\in I_{s^*}(v)}c_i(s)\leq 
\alpha\left(\sum_{i=1}^{|I_{s^*}(v)|}\frac{w_i}{W(I_{s^*}^i(v))}\right)
\left(\sum_{i\in I_{s^*}(v)}w_id(u_i,v)+\beta_v\right)\Rightarrow
\]

\begin{equation}
\sum_{i\in I_{s^*}(v)}c_i(s)\leq
\alpha\left(\sum_{i=1}^{|I_{s^*}(v)|}\frac{w_i}{W(I_{s^*}^i(v))}\right)
c_{I_{s^*}(v)}(s^*)
\label{equation:weighted-harmonic}
\end{equation}

\noindent The result will follow from~(\ref{equation:weighted-harmonic}), because $I_{s^*}(v)\subseteq A_{s^*}(v)$, and agents in $A_{s^*}(v)\setminus I_{s^*}(v)$ play $v$ in $s$ and $s^*$. We analyze the SPoA ratio. For $w_i=1$ it is $W(I_{s^*}^i(v))=|I_{s^*}^i(v)|=i$, which yields $SPoA=O(\alpha H(n))$. For weighted agents set $W(I_{s^*}^0(v))=0$ and use $x\geq 1+\ln x$ for $x\in (0,1)$ on~(\ref{equation:weighted-harmonic}):

\[
\alpha\sum_{i=1}^{|I_{s^*}(v)|}\Bigl(1-\frac{W(I_{s^*}^{i-1}(v))}{W(I_{s^*}^i(v))}\Bigr)\leq
-\alpha\sum_{i=1}^{|I_{s^*}(v)|}\ln\frac{W(I_{s^*}^{i-1}(v))}{W(I_{s^*}^i(v))}=
\alpha(1+\ln W)
\]
\qed
\end{proof}

\bigskip

\noindent {\bf Unweighted Lower Bound} We refer to Fig.~\ref{figure:apx-spoa-lb} for a tight (non-metric) lower bounding example in case of unweighted agents. Facility opening costs are $1$. It is clear that all agents being connected to node $v_{eq}$ is the most expensive $\alpha$-approximate equilibrium, of social cost $\alpha H(n)$: no agent coalition has incentive to deviate to $v_{opt}$, which is practically the node that all agents reside on, and results in a social cost of virtually $1$, for $\epsilon\rightarrow 0$ (e.g. $\epsilon=n^{-2}$). 

\begin{figure}[t]
\center
\includegraphics[scale=0.5]{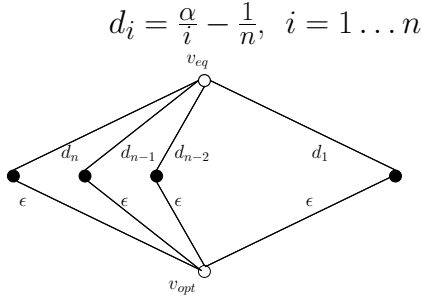}
\caption{Unweighted Non-metric Lower Bound for the SPoA of $\alpha$-approximate strong equilibria ($\alpha\geq e$).}
\label{figure:apx-spoa-lb}
\end{figure}

\subsection{Unweighted Agents on Metric Networks}
\label{subsection:metric-unweighted-spoa}

\noindent We prove the following for the SPoA of the unweighted metric Facility Location game:

\begin{theorem}
The Price of Anarchy of $e$-approximate strong equilibria in the  unweighted metric Facility Location game with fairly allocated facility costs is upper bounded by a constant.
\label{theorem:metric-spoa-ub}
\end{theorem}

The proof of theorem~\ref{theorem:metric-spoa-ub} consists of several partial results that we will synthesize. Let $s$ denote any $\alpha$-approximate strong equilibrium and $s^*$ the socially optimum cofiguration. We will upper bound the $SPoA$ again by $\max_{v\in F_{s^*}}\left(\sum_{i\in A_{s^*}(v)}c_i(s)\right)/c_v(s^*)\Bigr)$. For any facility node $v\in F_{s^*}$, define $A_{s}(v)\subseteq A_{s^*}(v)$ to be the subset of those agents that are connected to $v$ both in $s^*$ and $s$. Define $I_{s^*}(v)=A_{s^*}(v)\setminus A_{s}(v)$ to be the subset of agents that are connected to $v$ in $s^*$ but not in $s$. See fig.~\ref{figure:metric-spoa-setup} for an illustration of the definitions. To simplify notation we use $d(u_i,v)=x_i^*$, for $i\in A_{s^*}(v)$. At first we consider the following simple case:

\begin{figure}[t]
\center
\includegraphics[scale=0.7]{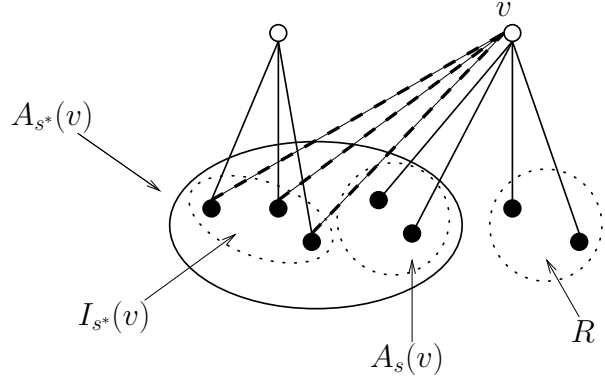}
\caption{The situation examined in the proof of theorem~\ref{theorem:metric-spoa-ub}.}
\label{figure:metric-spoa-setup}
\end{figure}

\begin{lemma}
Let $I_{s^*}(v)$ be a subset of misconnected agents under $\alpha$-approximate strong equilibrium profile $s$. Let $R$ denote agents connected to $v$ under $s$, with $R\subseteq A\setminus A_{s^*}(v)$. If  under $s_{-R}$ no agent of $I_{s^*}(v)$ has incentive to deviate to $v$ in coordination with $I_{s^*}(v)$, then $\sum_{i\in A_{s^*}(v)}c_i(s)\leq (1+\alpha) c_v(s^*)$.
\end{lemma}

\begin{proof}
If under $s_{-R}$ no agent $i\in I_{s^*}(v)$ has incentive to deviate in coordination with $I_{s^*}(v)\setminus\{i\}$ to $v$, then for every $i\in I_{s^*}(v)$ it is $c_i(s)\leq \alpha c_i(s^*)$, because $s$ is an $\alpha$-approximate strong equilibrium. Then:

\begin{equationarray}{lcl}
\sum_{i\in A_{s^*}(v)}c_i(s) & = &
\sum_{i\in A_s(v)}c_i(s)+\sum_{i\in I_{s^*}(v)}c_i(s)\nonumber\\
& \leq &
\beta_v+\sum_{i\in A_s(v)}x_i^*+
\alpha\left(
\beta_v+\sum_{i\in I_{s^*}(v)}x_i^*
\right)\nonumber\\
& \leq &
(1+\alpha)\left(\beta_v+\sum_{i\in A_{s^*}(v)}x_i^*\right)\nonumber
\end{equationarray}

\noindent The latter equals $(1+\alpha)c_v(s^*)$.\qed
\end{proof}

\bigskip

For the rest of the analysis we treat the complementary case of that described in the previous lemma. Assume there exists at least one agent $i\in I_{s^*}(v)$ willing to deviate to $v$ in coordination with the coalition $I_{s^*}(v)$ under $s_{-R}$. Define a {\em minimal disagreeing subset} $I_{s^*}^0(v)\subseteq I_{s^*}(v)$ to be a {\em minimal subset of misconnected agents} containing an agent $i$ that would actually deviate to $v$ with $I_{s^*}^0(v)$, under $s_{-R}$. Then also define $J_{s^*}(v)=I_{s^*}(v)\setminus I_{s^*}^0(v)$. Fix $i\in I_{s^*}^0(v)$ to be from now on the agent (or one of them if there are many) that would deviate to $v$ in coordination with $I_{s^*}^0(v)$. We call $i$ the {\em unstable} agent of the minimal disagreeing subset $I_{s^*}^0(v)$. By definition, the following holds for an {\em unstable} agent $i$ of the minimal disagreeing subset:

\begin{equation}
c_i(s) > \alpha\left(x_i^*+\frac{\beta_v}{|I_{s^*}^0(v)|+|A_s(v)|}\right)
\label{equation:i-lb}
\end{equation}

\noindent The rest of the analysis consists of bounds for agents in $I_{s^*}^0$ and $J_{s^*}(v)$ separately. In particular, lemmas~\ref{lemma:rest-ub} and~\ref{lemma:dev-ub} describe upper bounds for these sets respectively.

\begin{lemma}
Let $I_{s^*}(v)$ be a subset of misconnected agents under an $\alpha$-approximate strong equilibrium profile $s$. Let $I_{s^*}^0(v)$ be a minimal disagreeing subset of $I_{s^*}(v)$ and $i\in I_{s^*}^0(v)$ an unstable agent. Then:

\begin{equation}
\sum_{l\in I_{s^*}^0(v)}c_l(s)\leq 2\alpha\Bigl(
\beta_v+\sum_{l\in I_{s^*}^0(v)}x_l^*
\Bigr)
\label{equation:rest-ub}
\end{equation}
\label{lemma:rest-ub}
\end{lemma}

\begin{proof}
\noindent By minimality of $I_{s^*}^0(v)$, every agent $l\in I_{s^*}^0(v)$ is {\em not willing} to deviate to $v$ under $s_{-R}$ with a coalition of size $|I_{s^*}^0(v)|-1$. Thus for every $l\in I_{s^*}^0(v)$:

$$
c_l(s)\leq\alpha\left(x_l^*+\frac{\beta_v}{|I_{s^*}^0(v)|+|A_s(v)|-1}\right)
$$

\noindent Summing over $I_{s^*}^0(v)$ yields:

\[
\sum_{l\in I_{s^*}^0(v)}c_l(s)\leq\alpha\Bigl(
\sum_{l\in I_{s^*}^0(v)}x_l^*
+\frac{|I_{s^*}^0(v)|\beta_v}{|I_{s^*}^0(v)|+|A_s(v)|-1}
\Bigr)
\]

\noindent which is upper bounded by at most as in~(\ref{equation:rest-ub}).\qed
\end{proof}

\begin{lemma}
Let $I_{s^*}(v)$ be a subset of misconnected agents under an $\alpha$-approximate strong equilibrium profile $s$. Let $I_{s^*}^0(v)$ be a minimal disagreeing subset of $I_{s^*}(v)$ and $J_{s^*}(v)=I_{s^*}(v)\setminus I_{s^*}^0(v)$. Then:

\begin{equation}
\sum_{j\in J_{s^*}(v)}c_j(s)\leq\alpha\left[
\sum_{j\in J_{s^*}(v)}x_j^*+
\beta_v\Bigl(H(|A_{s^*}(v)|)-H(|I_{s^*}^0(v)|+|A_s(v)|)\Bigr)
\right]
\label{equation:dev-ub}
\end{equation}
\label{lemma:dev-ub}
\end{lemma}

\begin{proof}
\noindent 
Without loss of generality name agents $j\in J_{s^*}(v)$ by distinct indices $1,\dots, |J_{s^*}(v)|$ and define a series of supersets of $I_{s^*}^0(v)$, as follows: $I_{s^*}^j(v)=I_{s^*}^{j-1}(v)\cup\{j\}$. Because $s$ is an $\alpha$-approximate strong equilibrium, every set $I_{s^*}^j(v)$ contains an agent that is not willing to deviate to $v$ in coordination with $I_{s^*}^j(v)$. This agent is found either in $I_{s^*}^0(v)-\{i\}$ or in $I_{s^*}^j(v)\setminus I_{s^*}^0(v)$. We can assume without any loss of generality that for subset $I_{s^*}^j(v)$ this agent is $j$; otherwise we only need to exchange $j$ with some agent from $I_{s^*}^0(v)\setminus\{i\}$. One easily verifies that, by definition of a minimal disagreeing subset, such an exchange will not affect any of our previous results up to now. Then we have:

\begin{equation}
c_j(s)\leq \alpha\Bigl(x_j^*+\frac{\beta_v}{|I_{s^*}^j(v)|+|A_s(v)|+|R|}\Bigr)
\mbox{,\ \ }
j=1,\dots, |J_{s^*}(v)|\in J_{s^*}(v)
\label{equation:j-ub}
\end{equation}

\noindent We omit $|R|$ and sum the inequality over $j\in J_{s^*}(v)$. The result follows.\qed
\end{proof}

\bigskip

The following lemma will provide a lower bound for $\sum_{j\in J_{s^*}(v)}x_j^*$ appearing in~(\ref{equation:dev-ub}). Note that lemma~\ref{lemma:rest-ub} provides a concrete upper bound for the cost (under strategy profile $s$) of agents in $I_{s^*}(v)$, as a function of their connection cost in the socially optimum configuration $s^*$ and the corresponding facility cost $\beta_v$. This is not the case with lemma~\ref{lemma:dev-ub} for agents in $J_{s^*}(v)$; the upper bound~(\ref{equation:dev-ub}) involves socially optimum connection cost plus $\beta_v$ multiplied by additional terms of harmonic numbers. Thus we need to determine how low the socially optimum connection cost can be, in order to derive an upper bounding ratio for the $SPoA$.

\begin{lemma}
Let $I_{s^*}(v)$ be a subset of misconnected agents under $\alpha$-approximate strong equilibrium profile $s$. Let $I_{s^*}^0(v)$ be a minimal disagreeing subset of $I_{s^*}(v)$ and $J_{s^*}(v)=I_{s^*}(v)\setminus I_{s^*}^0(v)$. Then:

\begin{equation}
\sum_{j\in J_{s^*}(v)}x_j^*\geq
\frac{\beta_v}{1+\alpha}\Bigl(
\frac{|A_{s^*}(v)|-\lceil\alpha r\rceil}{r}-\alpha
\Bigl(H(|A_{s^*}(v)|)-H(\lceil\alpha r\rceil)
\Bigr)
\Bigr)
\mbox{,\ $r=|I_{s^*}^0(v)|+|A_s(v)|$}
\label{equation:spoa-connection-lb}
\end{equation}
\end{lemma}

\begin{proof}
Let $i$ be the fixed {\em unstable} agent of $I_{s^*}^0(v)$. Note that, under strategy profile $s$, $i$ does not have an incentive to join facility node $s_j$ for any $j\in J_{s^*}(v)$. Thus if $j$ pays for $s_j$ a share of $\frac{\beta_{s_j}}{\lambda_j}$ (that is, $s_j$ serves $\lambda_j$ agents in total in $s$):

\[
c_i(s)\leq
\alpha\left(
d(u_i, s_j)+\frac{\beta_{s_j}}{1+\lambda_j}\right)\leq
\alpha\left(
d(u_i, v)+d(u_j,v)+d(u_j,s_j)+\frac{\beta_{s_j}}{\lambda_j}\right)\Rightarrow
\]

\begin{equation}
c_i(s)\leq
\alpha\left(
x_i^*+x_j^*+c_j(s)
\right)
\leq
\alpha\left(
x_i^*+x_j^*+\alpha\left(
x_j^*+\frac{\beta_v}{|I_{s^*}^j(v)|+|A_s(v)|}\right)
\right)
\label{equation:everyone-ub}
\end{equation}

\noindent The latter inequality derives by usage of~(\ref{equation:j-ub}) for $c_j(s)$, and by safely omitting $|R|$. Using~(\ref{equation:everyone-ub}) and  the lower bound for $c_i(s)$ given in~(\ref{equation:i-lb}), we can solve for $x_j^*$. By the definition of $I_{s^*}^j(v)$ in lemma~\ref{lemma:dev-ub}, it is $|I_{s^*}^j(v)|=|I_{s^*}^0(v)|+j$, $j=1,\dots, |J_{s^*}(v)|$. Then we obtain:

\[
x_j^*\geq\max\Bigl\{0,
\frac{\beta_v}{1+\alpha}\Bigl(
\frac{1}{|I_{s^*}^0(v)|+|A_s(v)|}-
\frac{\alpha}{j+|I_{s^*}^0(v)|+|A_s(v)|}\Bigr)
\Bigr\}
\mbox{,\ \ }
j=1, \dots, |J_{s^*}(v)|
\]

\noindent Finally we sum up the latter bound over all $j$. Notice that $x_j^*$ becomes non-negative only when $j+|I_{s^*}^0(v)|+|A_s(v)|\geq \alpha (|I_{s^*}^0(v)|+|A_s(v)|)$. Since $j+|I_{s^*}^0(v)|+|A_s(v)|$ is an integral value, it turns out that $x_j^*$ becomes non-negative for those values of $j$ for which it is $j+|I_{s^*}^0(v)|+|A_s(v)|\geq \lceil\alpha (|I_{s^*}^0(v)|+|A_s(v)|)\rceil$. Then, by setting $r=|I_{s^*}^0(v)|+|A_s(v)|$, and by summing up over all $j$ we obtain the specified lower bound.\qed
\end{proof}

\paragraph{Proof of Theorem~\ref{theorem:metric-spoa-ub}}{
\noindent Now we put everything together. A {\em lower bound} on $c_v(s^*)$ is:

\[
c_v(s^*)\geq \beta_v+
\sum_{j\in J_{s^*}(v)}x_j^*+
\sum_{l\in I_{s^*}^0(v)}x_l^*+
\sum_{i\in A_s(v)}x_i^*
\]

\noindent Accordingly, we obtain the following {\em upper bound} on  $\sum_{i\in A_{s^*}(v)}c_i(s)$ by~(\ref{equation:rest-ub}):

\begin{equationarray}{lcl}
\sum_{i\in A_{s^*}(v)}c_i(s) 
& \leq &
\sum_{l\in I_{s^*}^0(v)}c_l(s)+
\sum_{j\in J_{s^*}(v)}c_j(s)+\sum_{i\in A_s(v)}c_i(s)\nonumber\\
& \leq &
2\alpha\Bigl(\sum_{l\in I_{s^*}^0(v)}x_l^*+\beta_v\Bigr)+
\sum_{j\in J_{s^*}(v)}c_j(s)+\sum_{i\in A_s(v)}c_i(s)\nonumber
\end{equationarray}

\noindent Using the latter bounds for $c_v(s^*)$ and $\sum_{i\in A_{s^*}(v)}c_i(s)$, we deduce:

\begin{equationarray}{lcr}
SPoA & \leq &
1+\Bigl(
2\alpha\beta_v+
2\alpha\sum_{l\in I_{s^*}^0(v)}x_l^*+
\sum_{j\in J_{s^*}(v)}c_j(s)
\Bigr)/\Bigl(
\beta_v+\sum_{l\in I_{s^*}^0(v)}x_l^*+\sum_{j\in J_{s^*}(v)}x_j^*
\Bigr)\nonumber\medskip\\
& \leq &
1+2\alpha\times\frac{\displaystyle
\beta_v+\sum_{j\in J_{s^*}(v)}x_j^*+\beta_v\left(H(|A_{s^*}(v)|)-H(r)\right)
}{\displaystyle
\beta_v+\sum_{j\in J_{s^*}(v)}x_j^*
}\nonumber\medskip\\
& = & 1+2\alpha +
2\alpha\times\frac{
\beta_v\left(H(|A_{s^*}(v)|)-H(r)\right)
}{\displaystyle
\beta_v+\sum_{j\in J_{s^*}(v)}x_j^*
}
\label{equation:metric-spoa-ub}
\end{equationarray}

\noindent The latter inequality emerged firstly by substitution from~(\ref{equation:dev-ub}), and secondly by removal of $\sum_{l\in I_{s^*}^0(v)}x_l^*$ from numerator and denominator. We maximize~(\ref{equation:metric-spoa-ub}) by taking the lower bound of the denominator, given in~(\ref{equation:spoa-connection-lb}) and obtain:

\begin{equation}
SPoA\leq 1+2\alpha+2\alpha\frac{
H(|A_{s^*}(v)|)-H(r)
}{
1+\frac{1}{1+\alpha}\Bigl(
\frac{|A_{s^*}(v)|-\lceil\alpha r\rceil}{r}-\alpha
\Bigl(
H(|A_{s^*}(v)|)-H(\lceil\alpha r\rceil)
\Bigr)
\Bigr)
}
\label{equation:analytic-metric-spoa-ub}
\end{equation}

\noindent Because $\lceil\alpha r\rceil\leq (\alpha+1)r$, and by using logarithmic bounds for the harmonic numbers:

\[
SPoA\leq 1+2\alpha+2\alpha\frac{
1-\gamma+\ln\frac{|A_{s^*}(v)|}{r}
}{
\frac{1}{1+\alpha}\Bigl(
\frac{|A_{s^*}(v)|}{r}-
\alpha\ln\frac{|A_{s^*}(v)|}{r}+
\alpha(\gamma+\ln\alpha-1)
\Bigr)
}
\]

\noindent By substituting $0.5$ for $\gamma > 0.5$ and $\alpha=e$ we can
maximize numerically the resulting upper bounding function of
$y=\frac{|A_{s^*}(v)|}{r}$ to at most a constant. Note that, when strong
equilibria exist, it is $\alpha=1$. Substituting so
in~(\ref{equation:analytic-metric-spoa-ub}) yields an upper bounding
expression similar to the one of PoS~(\ref{equation:metric-pos-ub}), up to
constant multiplicative and additive terms.\qed
}

\section{Open Problems}
\label{section:open-problems}

Further investigation of the weighted game is mostly challenging: it is not known whether this game possesses pure equilibria. Designing a counter-example seems quite demanding, as the game specializes in some sense the single-sink weighted network design game studied in~\cite{Chen06}. For this case pure equilibria were shown not to exist generally. Extending our analysis of the PoS and SPoA to the weighted metric case appears to be also non-trivial. Finally, derivation of good lower bounds for the PoS and the SPoA of the (non-metric) weighted case is an interesting aspect of research: lower bounding techniques developed  in~\cite{Albers08} do not readily apply for Facility Location.

\paragraph{Acknowledgements}{
For usefull discussions We thank the participants of the {\em Open Problems Jam Sessions} of CAGT: Daniel Andersson, Gudmund Skovb\-jerg Frand\-sen, Kristoffer Arnsfelt Hansen, Peter Bro Miltersen, Rocio Santillan Rodriguez,  Troels Bjerre Soerensen, Nikolaos Triandopoulos.
}

\end{document}